\documentclass[12pt,english,draftcls, one column]{IEEEtran}
\usepackage[T1]{fontenc}
\usepackage[latin9]{inputenc}
\usepackage{float}
\usepackage{amsmath}
\usepackage{amsthm}
\usepackage{amssymb}
\usepackage{graphicx}
\usepackage{multirow}
\usepackage{xcolor}
\usepackage{cite}

\usepackage{subcaption}

\makeatletter
\floatstyle{ruled}
\newfloat{algorithm}{tbp}{loa}
\providecommand{\algorithmname}{Algorithm}
\floatname{algorithm}{\protect\algorithmname}

\theoremstyle{plain}
\newtheorem{thm}{\protect\theoremname}
\theoremstyle{plain}
\newtheorem{prop}[thm]{\protect\propositionname}

\ifCLASSINFOpdf
\else
\fi

\usepackage{babel}

\makeatother

\usepackage{babel}
\providecommand{\propositionname}{Proposition}
\providecommand{\theoremname}{Theorem}

\begin{document}

\title{\LARGE{Deep Learning Methods for Joint Optimization of Beamforming and Fronthaul Quantization in\\ Cloud Radio Access Networks}}

\author{Daesung Yu, Hoon Lee, Seok-Hwan Park, and Seung-Eun Hong \vspace{-7mm}\thanks{This work was supported by the National Research Foundation (NRF) of Korea grants funded by the Ministry of Education under Grants 2019R1A6A1A09031717 and 2021R1A6A3A13046157 and by the Ministry of Science and ICT under Grants 2021R1C1C1006557 and 2021R1I1A3054575. This work was also supported by Institute of Information \& communications Technology Planning \& Evaluation (IITP) grant funded by the Korea government (MSIT) (No. 2018-0-01659, 5G Open Intelligence-Defined RAN (ID-RAN) Technique based on 5G New Radio). 

D. Yu and S.-H. Park are with the Division of Electronic Engineering and the Future Semiconductor Convergence Technology Research Center, Jeonbuk
National University, Jeonju, Korea (email: \{imcreative93, seokhwan\}@jbnu.ac.kr).

H. Lee is with the Department of Information and Communications Engineering, Pukyong National University, Busan, Korea (email: hlee@pknu.ac.kr).

S.-E. Hong is with the Future Mobile Communication Research Division, Electronics and Telecommunications Research Institute, Daejeon 34129, South Korea (email: iptvguru@etri.re.kr).}}
\maketitle
\begin{abstract}
Cooperative beamforming across access points (APs) and fronthaul quantization strategies are essential for cloud radio access network (C-RAN) systems. The nonconvexity of the C-RAN optimization problems, which is stemmed from per-AP power and fronthaul capacity constraints, requires high computational complexity for executing iterative algorithms. To resolve this issue, we investigate a deep learning approach where the optimization module is replaced with a well-trained deep neural network (DNN). An efficient learning solution is proposed which constructs a DNN to produce a low-dimensional representation of optimal beamforming and quantization strategies. Numerical results validate the advantages of the proposed learning solution.
\end{abstract}
\theoremstyle{theorem}
\newtheorem{theorem}{Theorem} 
\theoremstyle{proposition}
\newtheorem{proposition}{Proposition} 
\theoremstyle{lemma}
\newtheorem{lemma}{Lemma} 
\theoremstyle{corollary}
\newtheorem{corollary}{Corollary} 
\theoremstyle{definition}
\newtheorem{definition}{Definition}
\theoremstyle{remark}
\newtheorem{remark}{Remark}

\section{Introduction}
Cloud radio access network (C-RAN) has been regarded as a promising architecture for the next-generation wireless networks \cite{Simeone-et-al:JCN16}.
The C-RAN enables centralized signal processing by means of fronthaul links connecting central processors (CPs) and access points (APs). Due to the limited capacity of practical fronthaul channels, transmission strategies of the APs should be jointly designed along with the fronthaul interaction methods, i.e., the fronthaul quantization policies \cite{Peng-et-al:WC15}. There have been intensive studies on optimizing the performance of the C-RAN systems by iterative algorithms, e.g., transceiver design \cite{Park-et-al:TSP13,Yu-et-al:WCL19} and AP clustering \cite{Guo-et-al:JSAC16}. These traditional schemes, however, would not be implemented in practice due to their high computational complexity for executing iterative calculations.

Recent progresses on deep learning (DL) techniques have opened new research directions for developing low-complexity optimization methods in wireless networks \cite{Sun-et-al:TSP18, Zhang-et-al:TWC20, Hao-et-al:Access18, Kim-et-al:WCL20}. The basic idea is to replace optimization modules with deep neural networks (DNNs) which are trained in advance for optimizing the system performance. The complexity of trained DNNs are much lower than that of conventional iterative algorithms since DNN computations are carried out by simple matrix multiplications. Power control problems in interfering networks are investigated in \cite{Sun-et-al:TSP18}. Supervised learning approaches are presented which train DNNs to memorize solutions generated by existing weighted minimum mean squared error (WMMSE) algorithms. The authors in \cite{Zhang-et-al:TWC20} address multi-antenna beamforming optimization tasks through the supervised DL technique. DNNs are designed to learn the computations of handcraft beamforming optimization algorithms by exploiting the known optimal solutions. Although the time complexity can be reduced by the DNNs, their training steps need numerous samples of the optimal solutions obtained from the iterative algorithms, thereby increasing the training difficulty.

To address this issue, recent works \cite{Hao-et-al:Access18,Kim-et-al:WCL20}  have investigated unsupervised DL techniques which can identify efficient optimization strategies without any labels, i.e., the solutions of conventional algorithms. DNNs are trained to yield beamforming vectors that maximize the sum-rate performance under the sum transmit power constraint. It has been reported that, without the prior information of the optimal solutions, the unsupervised DL-based beamforming schemes could achieve the almost identical performance to those of existing locally optimum algorithms with much reduced complexity.

This letter proposes an unsupervised DL approach for the C-RAN systems by handling the joint optimization task of transmit beamforming and fronthaul quantization. Compared to existing DL studies \cite{Hao-et-al:Access18,Kim-et-al:WCL20} focusing on conventional cellular systems with the sum power constraint, the special nature the C-RANs imposes the per-AP power budget, the fronthaul capacity constraints, and additional optimization variables regarding the fronthaul quantization. These pose nontrivial challenges in designing efficient structure of DNNs suitable for the C-RAN architecture. Therefore, the conventional DL-based beamforming optimization methods cannot be straightforwardly applied to our scenario. 

To this end, we develop a structural learning process which constructs a DNN to always provide feasible beamforming vectors and fronthaul quantization policies. The proposed DNN generates intermediate variables that optimally recover the beamforming vectors. The quantization strategy is then determined by the learned beamforming solutions. As a result, the DNN can be trained in an unsupervised manner without the information of the optimal solutions. Numerical results validate the advantages of the proposed DL method.

The remainder of this letter is organized as follows. In Sec. \ref{sec:System-Model}, we describe a downlink C-RAN system, and beamforming and fronthaul quantization optimization problem is formulated under constraints on per-AP power and fronthaul capacity. The proposed DL method will be detailed in Sec. \ref{sec:Proposed-DL}. Then, advantages of the proposed DL method are validated via numerical results. Finally, we conclude this letter with discussion of future works in Sec. \ref{sec:Conclusion}.

\section{System model and Problem Definition\label{sec:System-Model}}

Consider a downlink C-RAN in which a CP communicates with $K$ single-antenna user equipments (UEs) by controlling $M$ single-antenna APs. Let $\mathcal{M} \triangleq \{1,\ldots,M\}$ and $\mathcal{K} \triangleq \{1,\ldots,K\}$ be the sets of APs' and UEs' indices, respectively. Each AP $i\in\mathcal{M}$ is connected to the CP through a fronthaul link of capacity $C$ in bit/symbol.
The received signal of UE $k\in\mathcal{K}$ is written as 
\begin{align}
    y_k = \mathbf{h}_k^H \mathbf{x} + z_k, \label{eq:received-signal-downlink}
\end{align}
where $\mathbf{h}_k\in\mathbb{C}^{M}$ denotes the channel from APs to UE $k$, $\mathbf{x}\in\mathbb{C}^{M}$ represents the signal vector transmitted by all APs, and $z_k\sim\mathcal{CN}(0,1)$ is the additive noise at UE $k$. The transmitted signal $\mathbf{x}$ is subject to per-AP power constraints expressed as
\begin{align}
    \mathtt{E}\left[|x_i|^2\right] \leq P, \, i\in\mathcal{M}, \label{eq:per-AP-power-constraint}
\end{align}
where $x_i$ is the $i$th element of $\mathbf{x}$ representing the signal radiated by AP $i$ and $P$ stands for the power budget at each~AP.

The CP generates the transmit signal vector $\mathbf{x}$ by employing a cooperative linear beamforming followed by fronthaul quantization \cite{Park-et-al:TSP13}. The transmitted signal $\mathbf{x}$ is then modeled~as
\begin{align}
\mathbf{x} = \sum\nolimits_{k\in\mathcal{K}} \mathbf{v}_k s_k + \mathbf{q},\label{eq:quantized-signal}
\end{align}
where $s_k\sim\mathcal{CN}(0,1)$ and $\mathbf{v}_k\in\mathbb{C}^{M}$ denote the data signal and beamforming vector for UE $k$, respectively, and $\mathbf{q} \in \mathbb{C}^{M}\sim\mathcal{CN}(\mathbf{0},\boldsymbol{\Omega})$ with covariance matrix $\boldsymbol{\Omega}\in\mathbb{C}^{M\times M}$ models the quantization noise vector independent of $\mathbf{x}$ under Gaussian test channel. We employ an independent fronthaul quantization scheme where each signal $x_{i}$ is individually compressed across APs $i\in\mathcal{M}$. Then, $\boldsymbol{\Omega}$ is given by a diagonal matrix. Let $\omega_{i}\geq0$ be the $i$th diagonal element of $\mathbf{\Omega}$, i.e., $\boldsymbol{\Omega} = \text{diag}(\{\omega_i\}_{i\in\mathcal{M}})$, which represents the quantization noise power for the fronthaul link toward AP $i$. Due to the limited fronthaul capacity $C$, the following constraint should be satisfied for successful decompression of $x_i$ at AP~$i$~\cite{Gamal-et-al:Book2011}.
\begin{align}
    \log_2\Big( 1 + \Big(\sum\nolimits_{k\in\mathcal{K}} |v_{k,i}|^2 \Big)/ \omega_i \Big) \leq C, \label{eq:fronthaul-capacity-constraint}
\end{align}
where $v_{k,i}$ indicates the $i$th element of $\mathbf{v}_k$.
Defining $\mathbf{v}\triangleq \{\mathbf{v}_k\}_{k\in\mathcal{K}}$ and $\boldsymbol{\omega}\triangleq \{\omega_i\}_{i\in\mathcal{M}}$, the achievable rate of UE $k$ $f_{k}(\mathbf{v},\boldsymbol{\omega})$ can be written as 
\begin{align}
    \!\!f_k\!\left(\mathbf{v}, \boldsymbol{\omega}\right)
    \!= \log_2\!\!\left(\!\!1 + \frac{|\mathbf{h}_k^H \mathbf{v}_k|^2}{ 1 + \mathbf{h}_k^H \boldsymbol{\Omega}\mathbf{h}_k \!\! + \!\sum_{l\in\mathcal{K}\setminus\{k\}} \! |\mathbf{h}_k^H \mathbf{v}_l|^2 }\!\!\right)\!. \label{eq:achievable-rate-UE-k}
\end{align}

We jointly optimize the beamforming vectors $\mathbf{v}$ and quantization noise powers $\boldsymbol{\omega}$ for maximizing the sum-rate performance $f(\mathbf{v},\boldsymbol{\omega})\triangleq\sum_{k\in\mathcal{K}}f_k(\mathbf{v},\boldsymbol{\omega})$ while satisfying the transmit power budget (\ref{eq:per-AP-power-constraint}) and fronthaul capacity constraints (\ref{eq:fronthaul-capacity-constraint}). 
In addition to the CSI $\mathbf{h}$, the constraints $P$ and $C$ are regarded as important system parameters that possibly vary at each transmission, thereby affecting the optimization procedure. The corresponding problem is formulated as
\begin{subequations}\label{eq:problem}
\begin{align}
\underset{\mathbf{v},\boldsymbol{\omega}}{\mathrm{max}}\,\,\,&f(\mathbf{v},\boldsymbol{\omega})\label{eq:problem-objective}\\
\mathrm{s.t.}\,\,\,\,\, & \sum\nolimits_{k\in\mathcal{K}} |v_{k,i}|^2+\omega_i\leq P,\,\,\,\,
i\in\mathcal{M},\forall P, \forall C,\label{eq:problem-power}%
\\ \,\,\,\,\, & \sum\nolimits_{k\in\mathcal{K}}|v_{k,i}|^2\leq \beta\omega_i,\,\,\,\,i\in\mathcal{M},\forall P, \forall C,\label{eq:problem-fronthaul}
\end{align}
\end{subequations}
where the per-AP power constraint (\ref{eq:problem-power}) is obtained by substituting \eqref{eq:quantized-signal} into (\ref{eq:per-AP-power-constraint}), and \eqref{eq:problem-fronthaul} comes from \eqref{eq:fronthaul-capacity-constraint} with a weight for consuming transmit power of $\boldsymbol{\omega}$ at each AP defined by  $\beta\triangleq2^{C}-1$. Both constraints should be achieved for any given $P$ and $C$ so that the resulting beamformer $\mathbf{v}$ and the quantization strategy $\boldsymbol{\omega}$ become feasible for arbitrary system configurations. It is not easy to find the globally optimum solution to \eqref{eq:problem} due to the nonconvex objective function \eqref{eq:problem-objective}. A locally optimal solution can be obtained by the WMMSE algorithm \cite{Yu-et-al:WCL19}, but its iterative nature results in high computational burden for practical C-RAN systems. 

To this end, we propose a low-complexity solution to \eqref{eq:problem} using DL techniques. Due to the absence of the optimal solution, instead of employing supervised learning approaches \cite{Sun-et-al:TSP18,Zhang-et-al:TWC20}, our focus is on identifying unsupervised DL framework, which can be implemented without the knowledge of the optimal solution of problem \eqref{eq:problem}. The DL-based beamforming schemes have been recently presented in \cite{Zhang-et-al:TWC20, Hao-et-al:Access18, Kim-et-al:WCL20}, for conventional cellular networks with co-located antennas. Due to the implicit assumption of the infinite fronthaul capacity $C=\infty$, the fronthaul quantization issue has not been addressed in designing DNN architecture and its training strategy. In the following sections, we develop a new DL method which tackles the intrinsic properties of the C-RAN systems, i.e., the per-AP power constraint and fronthaul capacity limitations.

\section{Proposed Deep Learning Method\label{sec:Proposed-DL}}

We first recast the original problem \eqref{eq:problem} into a {\em functional optimization} formulation \cite{DLiu:20} suitable for generalized learning for environment's status $\{\mathbf{h},P,C\}$. It transforms the target of the optimization into a function representing an optimization procedure. Any formulations with specified inputs and outputs can be refined to functional optimization tasks. Problem \eqref{eq:problem} can be viewed as an identification procedure of solutions $\mathbf{v}$ and $\boldsymbol{\omega}$ for arbitrary given channel $\mathbf{h}$ and system parameters $P$ and $C$. Such an input-output relationship can be captured by a functional operator $\{\mathbf{v}, \boldsymbol{\omega}\} = \mathcal{V}(\mathbf{h}, P, C)$. The operator $\mathcal{V}(\cdot)$ will be designed by a proper DNN. Substituting this into \eqref{eq:problem} yields the functional optimization expressed by
\begin{subequations} \label{eq:problem-stochastic}
\begin{align}
&\underset{\mathcal{V}(\cdot)}{\mathrm{max}}\ \mathtt{E}_{\mathbf{h}, P, C}[f(\mathcal{V}(\mathbf{h}, P, C))], \label{eq:problem-mapping-objective}\\
&\mathrm{s.t.}\ \eqref{eq:problem-power}\ \text{and}\ \eqref{eq:problem-fronthaul},\label{eq:problem-mapping}
\end{align}
\end{subequations}
where $\mathtt{E}_{X}[\cdot]$ accounts for the expectation operator over a random variable $X$. The equivalence between \eqref{eq:problem} and \eqref{eq:problem-stochastic} is mathematically verified in \cite{DLiu:20} and the references therein. Unlike the original problem \eqref{eq:problem} which focuses on identifying the solution variables $\mathbf{v}$ and $\boldsymbol{\omega}$ for a certain $\{\mathbf{h},P,C\}$, the functional optimization in \eqref{eq:problem-stochastic} addresses the expected sum-rate maximization rather than its instantaneous value. Consequently, by solving \eqref{eq:problem-stochastic}, a generic mapping rule $\mathcal{V}(\cdot)$ for arbitrarily given input $\{\mathbf{h},P,C\}$ can be obtained.

The remaining work is to design a proper DNN that approximates the intractable operator $\mathcal{V}(\cdot)$ successfully. A straightforward approach is to construct a DNN taking $\{\mathbf{h}, P, C\}$ and $\{\mathbf{v},\boldsymbol{\omega}\}$ as input and output, respectively. We refer to this scheme as a direct learning (DiLearn) method. The DNN can be readily trained to maximize the average sum-rate through the standard stochastic gradient descent (SGD) algorithm. However, the performance of the DiLearn approach has been shown to be poor in various beamforming optimization tasks \cite{Zhang-et-al:TWC20, Kim-et-al:WCL20} even without the fronthaul constraint. This is mainly stemmed from the difficulties of training a DNN with a large number of output variables and the absence of expert knowledge assisting the design of a DNN. In our case, the DiLearn needs to find $2MK+M$ real-valued output variables, which is quite large particularly when both $M$ and $K$ increase. This motivates us to investigate an appropriate DNN structure having much reduced output dimension for addressing \eqref{eq:problem-stochastic} efficiently.

\subsection{Optimal Solution Structure} \label{sub:optimal-structure}

To design an efficient DL architecture, this subsection studies special properties of the optimal beamforming and quantization noise power. The following proposition states that the optimal $\boldsymbol{\omega}$ can be retrieved from the beamforming $\mathbf{v}$. 
\begin{prop} \label{prop:virtual-power-constraint}
The solutions $\mathbf{v}$ and $\boldsymbol{\omega}$ are feasible for 
\eqref{eq:problem} if
\begin{subequations}\label{eq:prop1}
\begin{align}
&\omega_i=\frac{1}{\beta} \sum\nolimits_{k\in\mathcal{K}} |v_{k,i}|^2,  i\in\mathcal{M}, \label{eq:optimal-quantization-noise-power-for-given-beamformer}\\
&\sum\nolimits_{k\in\mathcal{K}} |v_{k,i}|^2 \leq \frac{P}{1+1/\beta}, i\in\mathcal{M}. \label{eq:virtual-power-constraint-feasible}
\end{align}
\end{subequations}
\end{prop}

\begin{proof}
We will show that $\mathbf{v}$ and $\boldsymbol{\omega}$ satisfying \eqref{eq:prop1} are indeed feasible for \eqref{eq:problem}. By substituting \eqref{eq:optimal-quantization-noise-power-for-given-beamformer} into \eqref{eq:problem-fronthaul}, it is easy to see that \eqref{eq:problem-fronthaul} is satisfied with equality. Also, the feasibility for \eqref{eq:problem-power} is shown by combining \eqref{eq:optimal-quantization-noise-power-for-given-beamformer} and \eqref{eq:virtual-power-constraint-feasible}, it follows $\omega_{i}\leq\frac{P}{1+\beta}$ resulting in
\begin{align}
    \sum\nolimits_{k\in\mathcal{K}}|v_{k,i}|^2+\omega_{i}\leq\frac{P}{1+1/\beta}+\frac{P}{1+\beta}=P.
\end{align}
We thus attain \eqref{eq:problem-power}. This completes the proof.
\end{proof}
Notice that, for a given $\mathbf{v}$, $\omega_i$ in \eqref{eq:optimal-quantization-noise-power-for-given-beamformer} is indeed optimal for \eqref{eq:problem} since the individual rate $f_{k}(\mathbf{v},\boldsymbol{\omega})$ in (\ref{eq:achievable-rate-UE-k}) is a monotonically decreasing function for each $\omega_{i}$. Therefore, the optimal $\omega_{i}$ is readily obtained from \eqref{eq:optimal-quantization-noise-power-for-given-beamformer} once the beamforming solution $\mathbf{v}$ is optimized. This implies that the corresponding DNN architecture can be designed to produce $\mathbf{v}$ only.

With the optimal $\boldsymbol{\omega}$ at hands, \eqref{eq:problem-power} and \eqref{eq:problem-fronthaul} can be combined into a sole constraint \eqref{eq:virtual-power-constraint-feasible}. As will be explained, this leads to a simple implementation of the proposed DNN. In addition, the left-hand side of (\ref{eq:virtual-power-constraint-feasible}) measures the beamforming power consumed at AP $i$. Therefore, (\ref{eq:virtual-power-constraint-feasible}) can be regarded as a virtual power constraint at AP $i$ compensating for the finite fronthaul capacity $C$. Based on this intuition, we present the following proposition which shows the optimal beamforming structure under the per-AP power constraints for arbitrary given fronthaul quantization processes.

\begin{prop} \label{prop:optimal-beamforming-structure}
Under the per-AP transmit power constraints (\ref{eq:virtual-power-constraint-feasible}), the optimal beamforming structure for a given $\boldsymbol{\omega}$
can be written by $\mathbf{v}_k=\sqrt{p_k}\mathbf{u}_k$, $k\in\mathcal{K}$, where $p_k$ and $\mathbf{u}_k \in\mathbb{C}^{M}$ with $||\mathbf{u}_k||^2 = 1$ stand for the transmit power and the beam direction for UE $k$, respectively. Here, $\mathbf{u}_k$ can be parameterized by $K+M$ nonnegative real numbers $\boldsymbol{\lambda}=\{\lambda_k\}_{k\in\mathcal{K}}$ and $\boldsymbol{\mu}=\{\mu_i\}_{i\in\mathcal{M}}$~as
\begin{equation}
\mathbf{u}_k=\frac{(\sum_{l\in\mathcal{K}}\lambda_l\mathbf{h}_l\mathbf{h}_l^H+\mathrm{diag}(\mathbf{\boldsymbol{\mu}}))^{-1}\mathbf{h}_k}{\|(\sum_{l\in\mathcal{K}}\lambda_l\mathbf{h}_l\mathbf{h}_l^H+\mathrm{diag}(\mathbf{\boldsymbol{\mu}}))^{-1}\mathbf{h}_k\|}, k\in\mathcal{K}.\label{eq:optimal-beamforming-structure}
\end{equation}
\end{prop}

\begin{proof}
The proof follows a similar procedure in \cite{Yu-Lan:TSP07}. For any given $\boldsymbol{\omega}$, the optimal $\mathbf{v}$ of problem \eqref{eq:problem} can be obtained by solving the following problem.
\begin{subequations}\label{eq:prop2-2}
\begin{align}
\underset{\mathbf{v}}{\mathrm{max}}\,\,\,&\sum_{k\in\mathcal{K}}\log_2\left(1+\frac{|\tilde{\mathbf{h}}_k^H\mathbf{v}_k|^2}{1+\sum_{l\in\mathcal{K}\backslash\{k\}}|\tilde{\mathbf{h}}_k^H\mathbf{v}_l|^2}\right)\\
\mathrm{s.t.}\,\,\,\,\, & \sum_{k\in\mathcal{K}} |v_{k,i}|^2\leq \tilde{P},\,\,\,\,
i\in\mathcal{M}.\label{eq:problem-power_2-3-1}%
\end{align}
\end{subequations}
where $\tilde{\mathbf{h}}_k=\mathbf{h}_k/\sigma_k$, $\sigma_k^2=1+\mathbf{h}_k^H\boldsymbol{\Omega}\mathbf{h}_k$, and $\tilde{P}=P/(1+1/\beta)$. Problem (\ref{eq:prop2-2}) can be interpreted as the sum-rate maximization problem for a multi-user downlink system with per-antenna power constraints and constant noise power across users addressed in \cite{Yu-Lan:TSP07}. According to \cite{Zhang-et-al:TWC20, Yu-Lan:TSP07}, the optimal beamforming solution for problem (\ref{eq:prop2-2}) has a structure of
\begin{align}
\mathbf{v}_k=\sqrt{p_k}\mathbf{u}_k,
\end{align}
where $p_k\geq0$ is the power allocated to UE $k$, and $\mathbf{u}_k$ is the beamforming direction for UE $k$ given as
\begin{align}\label{eq:beam_structure1}
\mathbf{u}_k=\frac{(\sum_{l\in\mathcal{K}}\tilde{\lambda}_l\tilde{\mathbf{h}}_l\tilde{\mathbf{h}}_l^H+\mathrm{diag}(\boldsymbol{\mu}))^{-1}\tilde{\mathbf{h}}_k}{\|(\sum_{l\in\mathcal{K}}\tilde{\lambda}_l\tilde{\mathbf{h}}_l\tilde{\mathbf{h}}_l^H+\mathrm{diag}(\boldsymbol{\mu}))^{-1}\tilde{\mathbf{h}}_k\|},\,\,k\in\mathcal{K},
\end{align}
with nonnegative real variables $\tilde{\boldsymbol{\lambda}}$ and $\boldsymbol{\mu}$.

Substituting $\lambda_k=\tilde{\lambda}_k/\sigma_k^2$ and $\tilde{\mathbf{h}}_k=\mathbf{h}_k/\sigma_k$ into the direction vector in (\ref{eq:beam_structure1}), we obtain \eqref{eq:optimal-beamforming-structure}. This completes the proof. 
\end{proof}

Proposition \ref{prop:optimal-beamforming-structure} identifies a low-dimensional representation of the optimal beamforming $\mathbf{v}$. It reveals that, for a given $\boldsymbol{\omega}$, the beamforming vectors can be efficiently retrieved from $2K+M$ real-valued parameters $\mathbf{p}\triangleq\{p_k\}_{k\in\mathcal{K}}$, $\boldsymbol{\lambda}$, and $\boldsymbol{\mu}$. Thus, we can further reduce the size of DNN such that it outputs only $2K+M$ nonnegative variables $\{\mathbf{p}, \boldsymbol{\lambda}, \boldsymbol{\mu}\}$. Combining this with Proposition 1, the optimal quantization noise variance $\boldsymbol{\omega}$ can also be recovered from $\{\mathbf{p}, \boldsymbol{\lambda}, \boldsymbol{\mu}\}$ by using \eqref{eq:optimal-quantization-noise-power-for-given-beamformer}. Therefore, compared to the DiLearn method, the number of output variables of DNN has been reduced from $2MK\!+\!M$ to~$2K\!+\!M$.

Proposition \ref{prop:optimal-beamforming-structure} only finds an alternative parameterization of the optimal solutions, but not the determination processes of the intermediate variables $\{\mathbf{p}, \boldsymbol{\lambda}, \boldsymbol{\mu}\}$. Classical optimization techniques cannot be straightforwardly applied to identify those parameters due to their highly coupled structure in \eqref{eq:optimal-beamforming-structure}. We address this issue by exploiting data-driven DL techniques. 


\subsection{Proposed DL Methods} \label{eq:proposed-structure-training}

\begin{figure}
\centering\includegraphics[width=0.7\linewidth]{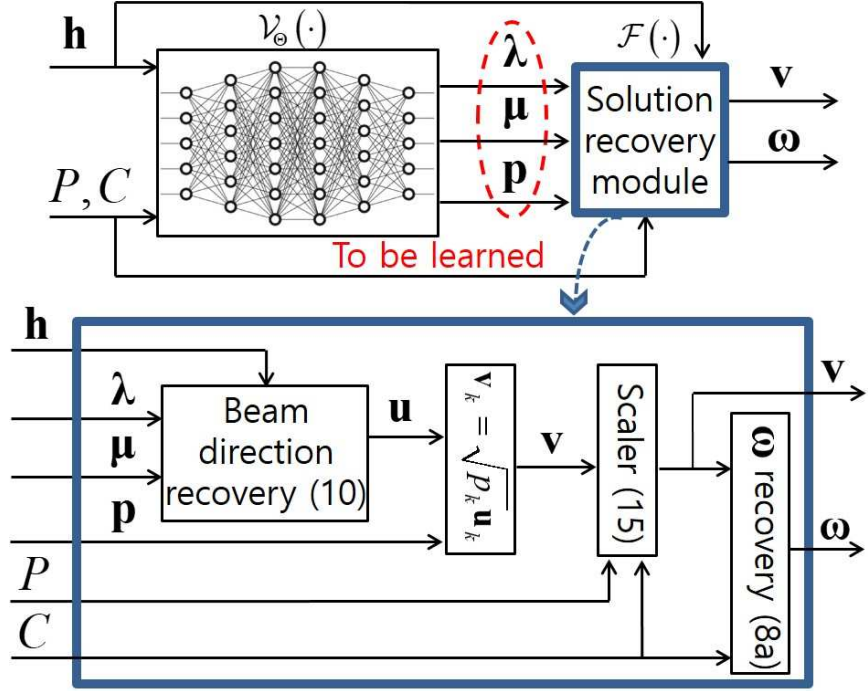}
\caption{{\label{fig:DL-structure}Proposed DL architecture}}
\end{figure}

Fig. \ref{fig:DL-structure} presents the proposed DL architecture which consists of two consecutive modules: DNN $\mathcal{V}_{\Theta}(\cdot)$ with trainable parameter $\Theta$ and solution recovery module $\mathcal{F}(\cdot)$. The training dataset contains numerous realizations of three-tuple $\{\mathbf{h},P,C\}$. The DNN accepts an input feature $\{\mathbf{h},P,C\}$ sampled from the training set and computes an output $\{\mathbf{p}, \boldsymbol{\lambda}, \boldsymbol{\mu}\}$, i.e., $\{\mathbf{p}, \boldsymbol{\lambda}, \boldsymbol{\mu}\}=\mathcal{V}_{\Theta}(\mathbf{h},P,C)$. For $l\in\mathcal{L}\triangleq\{1,\cdots,L\}$, the computation of layer $l$ is given as
\begin{align}
    \mathbf{d}_l = g_l\left( \mathrm{BN}\left( \mathbf{W}_l \mathbf{d}_{l-1} + \mathbf{b}_l \right)\right), \forall l\in\mathcal{L}, \label{eq:operation-hidden-layer}
\end{align}
where $g_l(\cdot)$ indicates the activation function for layer $l$, $\mathbf{W}_l\in\mathbb{R}^{S_{l}\times S_{l-1}}$ and $\mathbf{b}_{l}\in\mathbb{R}^{S_{l}}$ are weight matrix and bias vector, respectively, which collectively form the trainable parameter set $\Theta = \{\mathbf{W}_l, \mathbf{b}_l\}_{l\in\mathcal{L}}$. The batch normalization operation \cite{Ioffe-et-al:arXiv15} denoted by $\mathrm{BN}(\cdot)$ is included to accelerate the training step.
The final output of the DNN $\mathbf{d}_{L}$ of length $S_{L}=2K+M$ is represented by $\mathbf{d}_{L}=\{\mathbf{p}, \boldsymbol{\lambda}, \boldsymbol{\mu}\}$. The sequential calculations \eqref{eq:operation-hidden-layer} define the DNN mapping $\{\mathbf{p}, \boldsymbol{\lambda}, \boldsymbol{\mu}\}=\mathcal{V}_{\Theta}(\mathbf{h}, P, C)$.

The recovery module $\mathcal{F}(\cdot)$ further processes the DNN output $\{\mathbf{p}, \boldsymbol{\lambda}, \boldsymbol{\mu}\}$ to retrieve feasible solutions $\mathbf{v}$ and $\boldsymbol{\omega}$ as $\{\mathbf{v},\boldsymbol{\omega}\}=\mathcal{F}(\mathbf{p}, \boldsymbol{\lambda}, \boldsymbol{\mu})$. As illustrated in Fig. \ref{fig:DL-structure}, the beam direction vector $\mathbf{u}\triangleq\{\mathbf{u}_{k}\}_{k\in\mathcal{K}}$ is first obtained from the optimal structure \eqref{eq:optimal-beamforming-structure}, and then it is followed by pairwise multiplication $\mathbf{v}_{k}=\sqrt{p_{k}}\mathbf{u}_{k}$. To guarantee the feasibility of $\mathbf{v}$, we perform a simple scaling inspired by our analysis (\ref{eq:virtual-power-constraint-feasible}).
\begin{align}
    \mathbf{v} \leftarrow  \frac{\sqrt{P/(1 + 1/\beta)}}{\sqrt{\max_{i\in\mathcal{M}}\sum_{k\in\mathcal{K}} |v_{k,i}|^2}}   \mathbf{v}. \label{eq:scaling}
\end{align}
As discussed in Proposition \ref{prop:virtual-power-constraint}, the resulting $\mathbf{v}$ from \eqref{eq:scaling} becomes feasible to the original formulation \eqref{eq:problem}. The optimal quantization noise variance $\boldsymbol{\omega}$ is then computed according to (\ref{eq:optimal-quantization-noise-power-for-given-beamformer}). Finally, the proposed DL structure models the optimization function $\mathcal{V}(\cdot)$ in \eqref{eq:problem-mapping-objective} as
\begin{align}
    \mathcal{V}(\mathbf{h},P,C)=\mathcal{F}(\mathcal{V}_{\Theta}(\mathbf{h},P,C)).\label{eq:DNN}
\end{align}

Plugging this into \eqref{eq:problem-mapping-objective} results in a training problem written~by
\begin{align}
\underset{\Theta}{\mathrm{max}}\ \mathtt{E}_{\mathbf{h}, P, C}\big[f\big(\mathcal{F}(\mathcal{V}_{\Theta}(\mathbf{h},P,C))\big)\big].\label{eq:training}
\end{align}
Thanks to the scaling \eqref{eq:scaling}, both the transmit power and fronthaul capacity constraints in \eqref{eq:problem-mapping} can be lifted out in \eqref{eq:training}. The training problem \eqref{eq:training} can be readily addressed by the mini-batch SGD method, e.g., the Adam algorithm \cite{Kingma-et-al:ICLR15}. It iteratively updates the DNN parameter $\Theta$ by using the sample gradient evaluated over the mini-batch set $\mathcal{B}$ randomly sampled from the training dataset. The DNN parameter $\Theta^{[n]}$ obtained at the $n$th iteration is written by
\begin{align}
\Theta^{[n]} \! = \! \Theta^{[n-1]} \! - \! \gamma\mathtt{E}_\mathcal{B}[\bigtriangledown_{\Theta^{[n-1]}} f(\mathcal{F}(\mathcal{V}_{\Theta^{[n-1]}}(\mathbf{h},P,C)))],
\label{eq:update-rule}
\end{align}
where $\gamma>0$ denotes learning rate.

Unlike the supervised DL-based beamforming DNN \cite{Zhang-et-al:TWC20}, which relies on the optimal solutions generated from the iterative algorithms, the proposed training policy \eqref{eq:update-rule} does not require any prior knowledge of the nonconvex problem \eqref{eq:problem}, i.e., optimal $\mathbf{v}$, $\boldsymbol{\omega}$ of problem \eqref{eq:problem}. Thus, the proposed DL approach can be carried out in a fully unsupervised manner, resulting in a simple implementation of the training step. Notice that the training step is carried out in an offline manner before the real-time C-RAN deployment. Once the DNN is trained, the CP exploits the optimized parameter set $\Theta$ to calculate the solutions $\mathbf{v}$ and $\boldsymbol{\omega}$ from \eqref{eq:DNN} for new channel inputs.
\subsection{Complexity Analysis}\label{sub:complexity-analysis}
The proposed DL structure \eqref{eq:DNN} consists of matrix multiplications \eqref{eq:operation-hidden-layer} and beamforming recovery operation \eqref{eq:optimal-beamforming-structure}. We have found that about $13MK$ hidden neurons are sufficient for achieving a good performance-complexity trade-off. In this case, the overall time complexity of the DNN is given by $\mathcal{O}(M^2K^2+M^3)$. The WMMSE algorithm requires to solve convex semidefinite program repeatedly. Assuming $L_\text{WMMSE}$ iterations, the complexity of the WMMSE algorithm becomes $\mathcal{O}(L_\text{WMMSE}(MK+M)^{4.5})$, which is much higher than that of the proposed DL method. The complexity comparison will be numerically shown in Sec. \ref{sec:Numerical-Results}.

\section{Numerical Results\label{sec:Numerical-Results}}

This section provides numerical results validating the effectiveness of the proposed DL method. We consider $M=6$ APs and $K=6$ UEs uniformly distributed within a cell of radius 100 m. The one-ring channel model \cite{Yin-et-al:JSTSP14} is assumed, where there are single-scattering paths scattered by $N$ scatterers positioned on a disk-shaped scattering ring centered on the UE. Then, the channel vector of each UE $k$ is modeled as $\mathbf{h}_k=\sum_{n\in\{1,...,N\}}\mathbf{h}_{k,n}/{\sqrt{N}}$, where $\mathbf{h}_{k,n}=[\sqrt{\beta_{k,n,1}}e^{-j2\pi\frac{d_{k,n,1}+r}{\lambda_\text{c}}}\,\,...\,\,\sqrt{\beta_{k,n,M}}e^{-j2\pi\frac{d_{k,n,M}+r}{\lambda_\text{c}}}]^Te^{j\rho_{k,n}}$ with the path-loss between AP $i$ and UE $k$ via scatterer $n$ of UE $k$ $\beta_{k,n,i}$, distance between AP $i$ and scatterer $n$ of UE $k$ $d_{k,n,i}$, radius of scattering ring  $r$, common phase shift $\rho_{k,n}$ and wave length of carrier $\lambda_\text{c}$. Here, for all the scattering paths, $\beta_{k,n,i}$ have been defined as $\beta_{k,n,i}=1/(1+((d_{k,n,i}+r)/d_0)^\eta)$ with the reference distance $d_0$ and path-loss exponent $\eta$. For the simulations, we set the parameters as $d_0=30\,
\mathrm{m}$, $r=5\,\mathrm{m}$, $\eta=3$, $N=2$, $\lambda_\text{c}=0.15\,\mathrm{m}$, and $\rho_{k,n}\sim\mathcal{U}(0,2\pi)$ for $\forall k$, $\forall n$.
With the unit variances of the additive noises, the signal-to-noise ratio (SNR) is equal to $P$. A DNN is constructed with $L=11$ layers in which each hidden layer has $S_l = 480$ neurons. For hidden layers, we adopt the leaky rectified linear unit (LReLU) activation, which is given as $\mathrm{LReLU}(z)=z$ for $z\geq0$ and $\mathrm{LReLU}(z)=0.3z$ otherwise. To produce the nonnegative output $\{\mathbf{p}, \boldsymbol{\lambda}, \boldsymbol{\mu}\}$, the output layer is realized by the sigplus activation $\mathrm{SigPlus}(z)=\log(1+e^{z})$. The Adam optimizer \cite{Kingma-et-al:ICLR15} with the mini-batch size $B=10^4$ is employed as the SGD algorithm. The training step \eqref{eq:update-rule} proceeds until the validation performance is saturated. The trained DNN is evaluated with $100$ test samples.

\subsection{Dataset Generation}
The training samples $\{\mathbf{h}, P, C\}$ are randomly generated according to given distributions. As described, the channel vectors $\mathbf{h}$ follow the one-ring channel model, and the constraint factors $P$ and $C$ are sampled from the uniform distribution as $10\log_{10}P\sim\mathcal{U}( 10\log_{10}P_{\min}, 10\log_{10}P_{\max} )$ and $C \sim \mathcal{U}(C_{\min}, C_{\max})$, where the bounding parameters $(P_{\min}$, $P_{\max})= (1, 10^3)$ and $(C_{\min}$, $C_{\max})= (2, 10)$.
\vspace{-3mm}
\subsection{Results}

We compare the performance of the proposed DL approach with the following benchmark schemes: \textit{i)} WMMSE algorithm: A locally optimal solution to problem (\ref{eq:problem}) is found using the iterative WMMSE algorithm \cite{Yu-et-al:WCL19}; \textit{ii)} DiLearn: A DNN is designed to yield the beamforming vectors $\mathbf{v}$ directly. It is followed by the scaling operation in (\ref{eq:scaling}) and the computation of $\boldsymbol{\omega}$ in (\ref{eq:optimal-quantization-noise-power-for-given-beamformer}).

\begin{figure}
\centering
\begin{subfigure}[b]{0.8\textwidth}
\includegraphics[width=\textwidth]{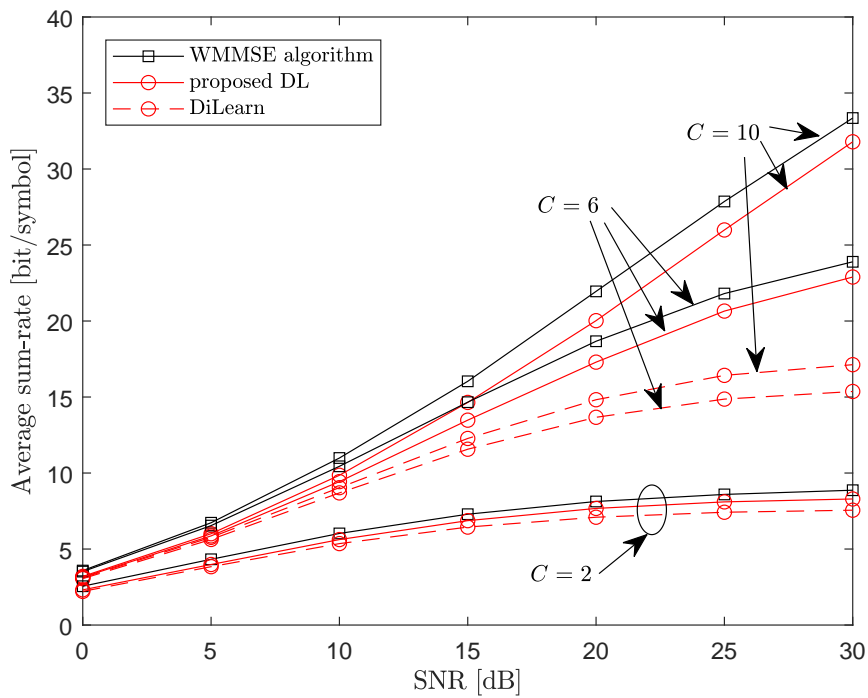}
\caption{Average sum-rate versus the SNR}\label{fig:vsSNR_MK66}
\end{subfigure}
\begin{subfigure}[b]{0.8\textwidth}
\includegraphics[width=\textwidth]{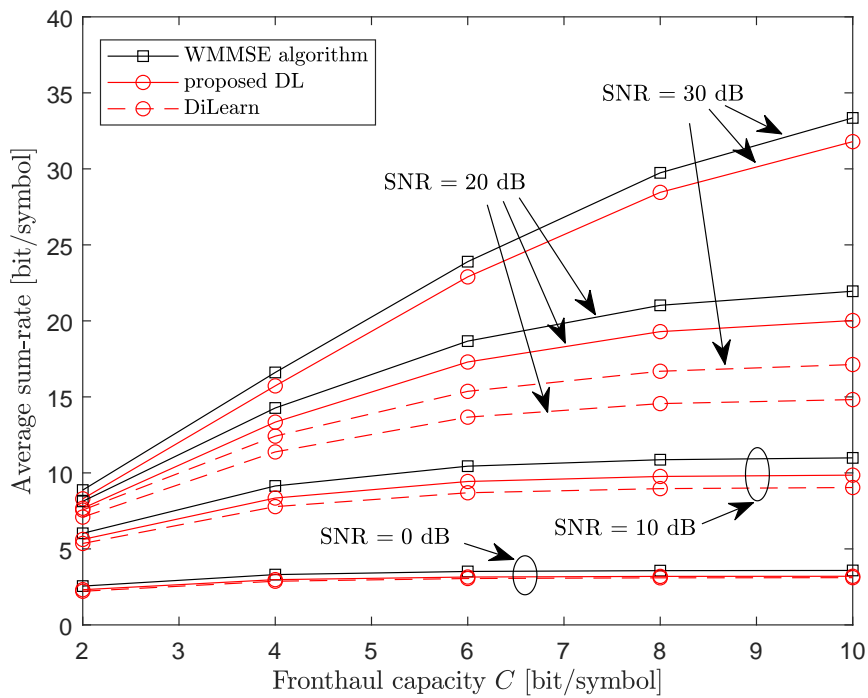}
\caption{Average sum-rate versus the fronthaul capacity $C$}\label{fig:vsC_MK66}
\end{subfigure}
\caption{Comparison of average sum-rate for $M=K=6$}\label{fig:MK66}
\end{figure}

In Fig. \ref{fig:MK66}, we evaluates the average sum-rate performance in the C-RAN with $M=K=6$. Fig. \ref{fig:vsSNR_MK66} depicts the average sum-rate performance by varying the SNR for $C\in\{2,\, 6,\, 10\}$. The proposed DL provides a good performance close to the WMMSE algorithm, whereas the DiLearn scheme exhibits severe performance loss. Similar observations can be made from Fig. \ref{fig:vsC_MK66} which presents the average sum-rate result with respect to the fronthaul capacity $C$ for $P\in\{0\,\,\mathrm{dB},10\,\,\mathrm{dB},20\,\,\mathrm{dB},30\,\,\mathrm{dB}\}$. The performance gap between the WMMSE and the DiLearn gets larger as $C$ and SNR grows. On the other hands, the proposed scheme shows only a slight loss compared to the WMMSE algorithm. This means that the DNN of the proposed scheme, which outputs only $2K+M=18$ variables, can be more efficiently trained than that of the DiLearn scheme whose output has $2MK=72$ variables.

\begin{table}[]
\caption{Average CPU run-time {[}sec{]} for $M=K=6$ with $C=10$}
\label{tab:CPU-time}
\hspace{+2mm}
\centering\begin{tabular}{|c|c|c|c|c|c|}
\hline
\multicolumn{4}{|c|}{WMMSE} & \multirow{2}{*}{proposed DL} & \multirow{2}{*}{DiLearn} \\ \cline{1-4}
0 dB & 10 dB & 20 dB & 30 dB &  &  \\ \hline
43.91 & 64.50 & 200.08 & 878.46 & 5.49$\times 10^{-3}$ & 5.29$\times 10^{-3}$ \\ \hline
\end{tabular}
\end{table}

Table \ref{tab:CPU-time} examines the advantage of the proposed DL scheme compared to the WMMSE algorithm in terms of the average CPU run-time at $C=10$. For the evaluations, both the trained DNNs and the WMMSE algorithm are implemented on a PC with an Intel i9-10900K CPU with 128 GB RAM using MATLAB R2020a. For $M=K=6$, the time complexity of the DL-based schemes is significantly lower than the WMMSE algorithm. Specifically, the gap between the WMMSE and DL-based schemes increases with SNR. This is because the WMMSE algorithm requires a larger number of iterations for convergence in the high SNR regime, while the DL-based schemes show the same complexity regardless of SNR as long as the DNN structures remain unchanged. The DiLearn scheme operates faster than the proposed scheme, since it does not require the matrix inversion in (\ref{eq:optimal-beamforming-structure}). However, the proposed scheme is more competitive considering the trade-off between the performance and complexity.

\section{Conclusions\label{sec:Conclusion}}
This letter has proposed DL methods for joint design of beamforming and fronthaul quantization strategies in C-RANs. The key idea is to design an efficient DNN architecture based on inherent relationships between optimal beamforming and quantization noise statistics.
Numerical results demonstrate that the proposed DL-based scheme achieves the best trade-off between the sum-rate performance and time complexity in comparison to baseline schemes.
As future works, a more generalized framework can be considered by including a channel learning process in the learning structure or considering generalization for the number of UEs and APs.  


\end{document}